\documentclass[12pt]{article}

\usepackage{amsfonts,amstext,amsmath,amssymb,amsthm,setspace,amsthm,mathtools,caption2,color}
\usepackage[authoryear]{natbib}
\usepackage{hyperref}

\def\({\left(}
\def\){\right)}
\newcommand{\bi}{\begin{itemize}}
\newcommand{\ei}{\end{itemize}}
\def\ba#1\ea{\begin{align}#1\end{align}}

\def\l{\left\{ }
\def\r{\right\} }

\newcommand{\vs}[1]{\vspace{#1 cm}}

\def\wh{\widehat}

\theoremstyle{plain} 
\newtheorem{theorem}{Theorem}
\newtheorem{lemma}{Lemma}

\theoremstyle{definition} 

\begin{document}
\title{\textbf{\Large Estimation of Gini Index within Pre-Specied Error Bound}}
\date{}
\vspace{-2cm}
\maketitle
\vskip -2.5cm
\author{
\vskip -2.0cm
\begin{flushleft}
\textbf{ \small BHARGAB CHATTOPADHYAY} \footnote{Corresponding author: Department of Mathematical Sciences, FO 2.402A, The University of Texas at Dallas, 800 West Campbell Road, Richardson, TX 75080, USA, E-mail: bhargab@utdallas.edu.}\\
\end{flushleft}
\vs{-.7}
\par \noindent\small Department of Mathematical Sciences, The University of Texas at Dallas, Richardson, Texas, USA. (bhargab@utdallas.edu)\\
\vspace{-.7cm}
\begin{flushleft}
 \textbf{ \small SHYAMAL KRISHNA DE} \\
\end{flushleft}
\vs{-.7}
\par \noindent\small
 School of Mathematical Sciences, National Institute of Science Education and Research, Bhubaneshwar, India. (sde@niser.ac.in)\\
}

\maketitle

\vspace{-1cm}
\begin{abstract}


Gini index is a widely used measure of economic inequality. This article develops a general theory for constructing a confidence interval for Gini index with a specified confidence coefficient and a specified width. Fixed sample size methods cannot simultaneously achieve both the specified confidence coefficient and specified width. We develop a purely sequential procedure for interval estimation of Gini index with a specified confidence coefficient and a fixed margin of error. Optimality properties of the proposed method, namely first order asymptotic efficiency and asymptotic consistency are proved. All theoretical results are derived without assuming any specific distribution of the data.

\noindent\textbf{KeyWords:} Fixed width confidence interval, Gini index, sample size planning, stopping rule, and U-statistics.\\
\noindent\textbf{Classification:} 62L12, 62G05, 60G46, 60G40 and 91B82.
\end{abstract}

\section{Introduction}
\label{s:intro}
Economic inequality arises due to the inequality in the distribution
of income and assets among individuals or groups within a society or region
or even between countries. Economic inequality is usually measured to
evaluate the effects of economic policies at the micro or macro level. In
the economics literature, there are several inequality indexes that measure the economic inequality. Among those indexes, Gini inequality index is the most widely used measure. The most celebrated Gini index, as given in \cite{arnold2005inequality}, is
arnold2005inequality
\begin{equation}
G_{F}(X)=\frac{\Delta }{2\mu }, \,\, \text{ where }\,\, \Delta =E\left\vert X_{1}-X_{2}\right\vert, \, \mu =E(X)
\label{def:Gini}
\end{equation}

and $X_{1}$ \& $X_{2}$ are two i.i.d. copies of non-negative random variable $X$. Gini index compares every individual's income with every other individual's income. If there are $n$ randomly selected individuals with incomes given by $X_{1}, \ldots,X_{n}$, then the estimator of the celebrated Gini index is
\begin{align}
\widehat{G}_n=\frac{\widehat{\Delta}_{n}}{2\bar{X}_{n}},
\label{def:Gnest}
\end{align}
where $\bar{X}_{n}$ is the sample mean and $\widehat{\Delta}_{n}$ is the sample Gini's mean difference defined as,
\begin{align}
\bar{X}_{n} = \frac{1}{n}\sum\limits_{i=1}^n X_i\text{ and }\widehat{\Delta}_{n}=\binom{n}{2}^{-1}\sum\limits_{1\leq i_{1}<i_{2}\leq n}\left\vert
X_{i_{1}}-X_{i_{2}}\right\vert.
\label{def:gmd}
\end{align}
The Gini index is undefined if $\bar{X}_{n}=0.$ We ignore this special case. 

For continuous evaluation of different economic policies implemented by the government, computation of Gini index for the whole country or a region is very important. One source of income or expenditure data for all households in a region is census data which is typically collected every 10 years. As a result, Gini index computed based on census data is available only once in every 10 years\footnote{For some countries, Gini indexes reported are based on even more than 15 years old data (see World Bank website). Examples include Belize, Algeria, and Botswana whose Gini indexes are based on data collected in 1999, 1995, and 1994 respectively}. Unless the household income data is annually updated, an estimate of Gini index for intermediate periods between two censuses can not be obtained. Some developed countries conduct household surveys annually\footnote{For instance, European Statistics on Income and Living Conditions conducts a household survey that collects data from at least 273000 individuals from each country in the European Union (see \url{http://epp.eurostat.ec.europa.eu/cache/ITY\_SDDS/EN/ilc\_esms.htm\#data\_rev})}. However, many countries can not afford or do not conduct household survey annually. For those countries, it is useful to draw relatively small number of households to estimate the population Gini index. In order to estimate Gini index for a region, a simple random sampling of households may be used\footnote{For estimating Gini index for smaller countries such as San Marino, Monaco etc., one can conduct small scale surveys by simple random sampling of households.}. For the existing literature on the use of simple random sampling to estimate inequality indexes, we refer to \cite{gastwirth1972estimation}, \cite{bishop1997statistical}, \cite{xu2007} and \cite{davidson2009reliable}.

In the economics literature, there exist innovative methods for constructing confidence intervals for $G_F$ (for e.g., see \cite{xu2007}). However, we know that the confidence interval varies from sample to sample and so is its width. Wider confidence intervals provide less precise information about the true value of the parameter of interest. Since it is desirable to construct shorter confidence intervals, we rather fix the length of the confidence interval, or in other words, the margin of error while achieving the same confidence coefficient. Thus we want to construct a $100(1-\alpha)\%$ fixed-width confidence interval for $G_F$. This problem is know as the fixed-width confidence interval estimation problem. 

No fixed sample size procedure can provide a solution to the fixed-width confidence interval estimation problem (e.g., see \cite{dant1940}). This problem falls in the domain of sequential analysis. For the details about the general theory of fixed-width confidence interval estimation, we refer the interested readers to \cite{sen1981sequentialbook} and \cite{mukhopadhyay2009sequential}. Sequential analysis is concerned with studies where sample sizes are not fixed in advance unlike fixed-sample size procedures. Instead, the sequential estimation procedure depends on collecting observations until an a-priori specified criterion or \textit{stopping rule} is satisfied.

We know that Gini's mean difference is U-statistic with a symmetric kernel of degree 2 and the sample mean is a U-statistic with a symmetric kernel of degree 1 (for e.g., see Hoeffding, \cite{hoeffding1948class}). Under distribution-free scenario, \cite{xu2007} used the central limit theorem for U-statistics to come up with a confidence interval for Gini index. However, this cannot be used to find out a fixed-width confidence interval for Gini index. In this article, we solve the problem of obtaining a fixed-width confidence interval for Gini index using a purely sequential procedure with a stopping rule based on several U-statistics. Apart from being unbiased estimators, U-statistics are also reverse martingales with respect to some non-increasing filtration as proven in \cite{lee1990u}. For more literature on reverse martingales, we refer to classical textbooks on probability theory and stochastic processes such as \cite{loeveprobability},  \cite{doob1953stochastic}, and others. We exploit the reverse martingale property of U-statistics to derive attractive asymptotic properties of our proposed estimation procedure.

In the next section, we formally state the fixed-width confidence interval estimation problem and why a fixed-sample size procedure cannot be used. In section 3, a purely sequential procedure is proposed to construct a $100(1-\alpha)\%$ fixed-width confidence interval for unknown population Gini index and  implementation and characteristics of the sequential procedure is discussed as well.  Section 4 presents simulation study and validate all theoretical results related to our procedure. We conclude this article with some remarks in section 5. 

\section{Problem Statement and Optimal Sample Size}
Consider $n$ randomly selected individuals from some population of interest with incomes denoted by $X_{1}, X_{2},\ldots,X_{n}$. These are nonnegative random variables. A strongly consistent estimator of population Gini index $G_F$ is $\wh{G}_n$ given in \eqref{def:Gnest}. For fixed $\alpha \in (0, 1)$, the goal of this paper is to develop the theory for constructing a $100(1-\alpha)\%$ fixed-width confidence interval for $G_F$. Formally, we would like to construct a confidence interval $J_n = (\wh{G}_n - d, \wh{G}_n +d)$ such that
\begin{align}
\label{fixed-width}
P\( \wh{G}_n - d < G_F < \wh{G}_n + d \) \ge 1-\alpha,
\end{align}
for some prefixed margin of error $d>0$.
Using \citet{xu2007}, we have
\begin{align}
\sqrt{n}\left(\widehat{G}_n-G_F\right)\to N\left(0, \xi^2\right) \,\,\,\, \text{ as } n\to \infty,
\label{clt}
\end{align}
where $\xi^2$ is the asymptotic variance given by,
\begin{align}
\xi^2 =\frac{\Delta ^{2}}{4\mu ^{4}}\sigma ^{2}-\frac{\Delta\tau}{\mu^{3}}+\frac{\Delta^2}{\mu^2}+\frac{\sigma _{1}^{2}}{\mu^{2}}.%
\label{def:xi}
\end{align}
Here,
\begin{align*}
\tau=E(X_{1}\left\vert X_{1}-X_{2}\right\vert)\,\, \text{ and }\,\, \sigma _{1}^{2}=V[E\left\vert X_{1}-X_{2}\right\vert {\vert}X_{1}=x_{1}].
\end{align*}
Based on the asymptotic normality of $\wh{G}_n$, we observe that the coverage probability is
\[
P\( \wh{G}_n - d < G_F < \wh{G}_n + d \) \approx 2 \Phi\( \frac{d \sqrt{n}}{\xi}\)-1,
\]
where $\Phi$ is the distribution function of standard normal random variable. In order to have $100(1-\alpha) \%$ confidence interval, sample size $n$ must satisfy
\begin{align}
\label{approx-cov-prob}
2 \Phi\( \frac{d \sqrt{n}}{\xi}\) -1 \ge 1-\alpha.
\end{align}
  Solving \eqref{approx-cov-prob} for $n$, we obtain $n \ge d^{-2}z_{\alpha /2}^{2}\xi ^{2}$, where $z_{\alpha /2}$ is the upper $\(\frac{\alpha}{2}\)^{th}$ quantile of the standard normal distribution. Thus, the optimal (minimal) sample size required to construct a fixed-width confidence interval for Gini index with approximately $(1-\alpha)$ coverage probability is 
\begin{align}
\label{optimal}
C= \lfloor d^{-2}z_{\alpha /2}^{2}\xi ^{2} \rfloor +1,
\end{align}
provided $\xi$ is known.

The optimal fixed sample size $C$ is unknown since the true value of $\xi $ is unknown in practice. If $C$ were known, one would just draw $C$ observations independently from the population of interest and compute $(\wh{G}_C - d , \wh{G}_C + d)$ which would satisfy \eqref{fixed-width} approximately. Since $C$ is unknown, one must draw samples at least in two stages in order to achieve the desired coverage probability at least approximately. In the first stage, one must estimate $C$ by estimating $\xi$, and then in the subsequent stages one should collect samples until the current sample size is more or equal to the estimated optimal sample size. In this article, we propose a sequential sampling procedure to estimate the optimal sample size $C$ and ensure that the fixed-width confidence interval based on the final sample size attains the desired $(1-\alpha)$ coverage probability.
%
%
%
\section{The Sequential Estimation Procedure}
\label{s:procedure}
In sequential estimation procedures, the parameter estimates are updated as the data is observed. In the first step, a small sample, called the pilot sample, is observed to gather preliminary information about the parameter of interest. Then, in each successive step, one or more additional observations are collected and the estimates of the parameters are updated. After each and every step a decision is taken whether to continue or to terminate the sampling process. This decision is based on a pre-defined stopping rule.

From (\ref{optimal}) we note that the optimal sample size needed to find a fixed-width confidence interval depends on unknown parameter $\xi^2$. So, let us first find a good estimator of the unknown parameter $\xi^2$. Following \cite{xu2007} and \cite{sproule1969sequential}, we consider the following strongly consistent estimator of $\xi^2$ based on U-statistics. Let us define a U-statistic, for each $j=1,2,\dots,n$,
\begin{align}
\widehat{\Delta }_{n}^{(j)}=\binom{n-1}{2}^{-1}\sum\limits_{T_{j}}%
\left\vert X_{i_1} - X_{i_2} \right\vert,
\end{align}
where ${T}_{j}{ =\{(i}_{1}{ ,i}_{2}{ ):1\leq i}_{1}%
{<i}_{2}{\leq n}$ and ${i}_{1}{,i}_{2}{ %
\neq j\}}$. Define $ W_{jn}{ =n}\widehat{{\Delta }}_{n}{-(n-2)}
\widehat{{\Delta }}_{n}^{(j)}$ for $j=1,\ldots,n$, and $\overline{W}_n = n^{-1}\sum_{j=1}^n W_{jn}$. According to \cite{sproule1969sequential}, a strongly consistent estimator of $4\sigma_1^2$ is
\[
{s}_{{wn}}^{{ 2}}={(n-1)}^{-1}\sum \limits_{i=1}^{n}( W_{jn} -\overline{W}_n )^2.
\]
Using \cite{xu2007},
\begin{align}
\widehat{\tau }_{n}=\frac{2}{n(n-1)}\sum\limits_{(n,2)}\frac{1}{2}(%
{ X}_{i_{1}}{ +X}_{i_{2}}{ )}\left\vert { X}_{i_{1}}%
{ -X}_{i_{2}}\right\vert
\end{align}%
is an estimator of $\tau$. Let $S_n^2$ be the sample variance. Thus, the estimator of ${\xi }^{%
{ 2}}$\ is%
\begin{equation}
{ V}_{{ n}}^{{ 2}}=\frac{\widehat{\Delta }%
_{n}^{2}S_{n}^{2}}{4\overline{X}_{n}^{4}}-\frac{\widehat{\Delta }_{n}}{%
\overline{X}_{n}^{3}}\widehat{\tau }_{n}+\frac{\widehat{\Delta }_{n}^{2}}{%
\overline{X}_{n}^{2}}+\frac{s_{wn}^{2}}{4\overline{X}_{n}^{2}},
\label{est-of-xi2}
\end{equation}
similar to \cite{xu2007}.
Using \cite{sproule1969sequential} and theorem 3.2.1 of \cite{sen1981sequentialbook}, we conclude that $V_{n}^{2}$ is a strongly consistent estimator of $\xi ^{2}$.
Based on this estimator of $\xi^2$, we define the stopping rule $N_d$, for every $d>0$, as
\begin{equation}
N_{d}\, \text{ is the smallest integer }\, n(\geq m)\text{ such that }n\geq \left(\frac{z_{\alpha/2}}{d}\right)^2\left( V_{n}^2+n^{-1}\right).
\label{stopping-rule}
\end{equation}
Here, $m$ is called the initial or pilot sample size, and the term $n^{-1}$ is known as a correction term. Note that $V_n$ can be very close to zero with positive probability. Without the correction term, the inequality \eqref{stopping-rule} may be satisfied for very small $n$ terminating the sampling process too early. Thus the correction term $n^{-1}$ ensures that the sampling process for estimating the optimal sample size does not stop too early. For details about the correction term, we refer to \cite{sen1981sequentialbook}.

From (\ref{stopping-rule}), we note that, $N_d\geq \left(\frac{z_\alpha/2}{d}\right)^2N_d^{-1}$, i.e., the final sample size must be at least $z_{\alpha/2}/d$. Therefore, we consider the pilot sample size to be $m=\max{\{4,z_{\alpha/2}/d\}}$. This technique of estimating pilot sample size can also be found in \cite{mukhopadhyay2009sequential}.

Recall that the optimal sample size required to achieve $100(1-\alpha)\%$ confidence interval for Gini index is $C$ which is unknown in practice. The stopping variable $N_d$ defined in \eqref{stopping-rule} serves as an estimator of $C$. Below, we develop a purely sequential procedure to estimate the optimal sample size $C$. 
\subsection{Implementation and Characteristics}
\label{results}
 We propose the following purely sequential estimation procedure to estimate $C$:\\
\textbf{Stage 1:} Compute the pilot sample size $m=\max{\{4,z_{\alpha/2}/d\}}$ and draw a random sample of size $m$ from the population of interest. Based on this pilot sample of size $m$, obtain an estimate of $\xi^2$ by finding ${V}_{{m}}^{{2}}$ as given in (\ref{est-of-xi2}) and check whether $m\geq (z_{\alpha/2}/d)^2\left(V_{m}^2+m^{-1}\right)$. If $m < (z_{\alpha/2}/d)^2\left(V_{m}^2+m^{-1}\right)$ then go to the next step. Otherwise, set the final sample size $N_d=m$.\\
\textbf{Stage 2:} Draw an additional observation independent of the pilot sample and update the estimate of $\xi^2$ by computing ${V}_{m+1}^{2}$. Check if $m+1\geq (z_{\alpha/2}/d)^2\left(V_{m+1}^2+(m+1)^{-1}\right)$. If $m+1 < (z_{\alpha/2}/d)^2\left(V_{m+1}^2+(m+1)^{-1}\right)$ then go to the next step. Otherwise, if $m+1 \geq (z_{\alpha/2}/d)^2\left(V_{m+1}^2+(m+1)^{-1}\right)$ then stop further sampling and report the final sample size as $N_d=m+1$.

This process of collecting one observation in each stage after stage 1 is continued until there are $N_d$ observations such that $N_d\geq (z_{\alpha/2}/d)^2\left(V_{N_d}^2+{N_d}^{-1}\right)$. At this stage, we stop sampling and report the final sample size as $N_d$.

Based on the above algorithm, the sampling process will stop at some stage. This is proved in Lemma 1 which states that if observations are collected using (\ref{stopping-rule}), under appropriate conditions, $P(N_d<\infty)=1$. This is a very important property of any sequential procedure since it mathematically ensures that the sampling will be terminated eventually.

Next, we establish some desirable asymptotic properties of our proposed sequential procedure. First, we prove that the final sample size $N_d$ required by our sampling strategy is close to the optimal sample size C at least asymptotically. This property is known as asymptotic efficiency property of sequential procedure which ensures that, on average, we collect only the minimum number of samples to achieve certain accuracy of estimation. Second, we show that the fixed-width confidence interval $\left(\widehat{G}_{N_d}-d,\widehat{G}_{N_d}+d\right)$ contains the true value of Gini index $G_F$ nearly with probability $1-\alpha$. We formally state these results in theorems 1 and 2.

\begin{theorem} If the parent distribution $F$ is such that E$[X^{4}]$ and E$[X^{-\beta}]$ exist for $\beta>4$,\footnote{If for a certain distribution function, negative moments doesn't exist, then theorem 1 will hold, if ${E}\left[\underset{n\geq m}{\sup } \,s_{wn}^{2}\overline{X}_{n}^{-2}\right]$, ${ E}\left[\underset{n\geq m}{\sup } \widehat{\tau }_{n}\overline{X}_{n}^{-2}\right]$ and ${ E}\left[ \underset{n\geq m}{\sup }S_{n}^{2}\overline{X}_{n}^{-2}\right]$ are finite.}  then the stopping rule in (\ref{stopping-rule}) yields the following asymptotic optimality properties:
\begin{itemize}
\item[(i)] $N_{d}/C\overset{a.s.}{\to} 1$ as $d \downarrow 0$.
\item[(ii)] $E\left(N_{d}/C\right)\to 1$ as $d \downarrow 0$.
\end{itemize}
\label{thm:main} 
\end{theorem}
\begin{theorem} If the parent distribution $F$ is such that E$[X^{4}]$ exist, then the stopping rule in (\ref{stopping-rule}) yields 
\label{thm:2}
\begin{align}
P\left(\widehat{G}_{N_{d}}-d<G_F<\widehat{G}_{N_{d}}+d\right)\to 1-\alpha\text{ as }d \downarrow 0.
\end{align}
\end{theorem}
Theorems 1 and 2 are proved in the appendix. Part (i) of theorem 1 implies that the ratio of final sample size of our procedure and the optimal sample size, C asymptotically converges to 1. Part (ii) of theorem 1 implies that the ratio of the average final sample size of our procedure and C asymptotically converges to 1. This property is called first order asymptotic efficiency property as it can be found in \cite{mukhopadhyay2009sequential}. Theorem 2 implies that the coverage probability produced by the fixed-width confidence interval $\left(\widehat{G}_{N_d}-d,\widehat{G}_{N_d}+d\right)$ attains the desired level $1-\alpha$ asymptotically. This property is called asymptotic consistency. Thus, we prove that the proposed purely sequential procedure enjoys both asymptotic efficiency property and asymptotic consistency property.

\section{Simulation Study}

In this section, we validate the asymptotic properties of our method stated in theorems 1 and 2 through Monte Carlo study. To implement the sequential procedure, we fix $d (=0.01)$ and $\alpha (=0.1)$. Using the pilot sample size formula $m=\max{ \{ 4, z_{\frac{\alpha}{2}}/d\} }$, the pilot sample size considered here is 165. Then, we implement the sequential procedure described in section \ref{results} and estimate the average sample size ($\overline{N}$), the maximum sample size ($\max (N)$), the standard error ($s(\overline{N})$) of $\overline{N}$, the coverage probability ($p$), and its standard error ($s_p$) based on $2000$ replications by drawing random samples from gamma distribution (shape $=2.649$,rate $=0.84$), log-normal distribution (mean $=2.185$, sd $=0.562$), and Pareto (20000, 5). Table \ref{tab1} summarizes the numerical results obtained from the simulation study. The parameters of log-normal and gamma distributions are same as used by Ransom and Cramer (\citeyear{ransom1983income}). 
%

From the fourth column of table \ref{tab1}, we  find that the ratio of the average final sample size and C is close to 1. Moreover, column 6 of table \ref{tab1} illustrates that the attained coverage probability is very close to the desired level of $90\%$. Thus, we find that the simulation results validate all theoretical results mentioned in the previous section, and the performance of the procedure is satisfactory for the above mentioned distributions. 
%
\section{Concluding Remarks}
Gini index is a widely used measure of economic inequality index. In order to evaluate the economic policies adopted by a government, it is important to estimate Gini index at any specific time period. If the income data for all households in the region of interest is not available, one should estimate Gini index by drawing a simple random sample of households from that region. This article develops a purely sequential procedure that provides a $100(1-\alpha)\%$ fixed-width confidence interval for Gini index. Without assuming any specific distribution for the data, we show that the ratio of the final sample size and the optimal sample size approaches 1. We also show that the confidence interval constructed using our proposed sequential method attains the required coverage probability. Thus, based on these results, we conclude that the proposed sequential estimation strategy can efficiently construct  a $100(1-\alpha)\%$ fixed-width confidence interval for Gini index. In this article, we consider that after pilot sample, one additional observation is collected in each step. If instead, a group of $r (\geq 1$, say) observations are collected in each step after the pilot sample stage, the same properties will hold. The proofs will be similar to the ones in Appendix.

Apart from economics, there are other fields where researchers report Gini index. For instance, in social sciences and economics, the Gini index is used to measure inequality in education (see \cite{thomas2001measuring}). In ecology, the Gini index is used as a measure of biodiversity (for e.g., see \cite{wittebolle2009initial}). \cite{asada2005assessment} uses Gini index as a measure of the inequality of health related quality of life in a population. \cite{shi2003greedy} uses Gini index to evaluate the fairness achieved by internet routers in scheduling packet transmissions from different flows of traffic. Possible application of Gini index in so many fields such as sociology, health science, ecology, engineering, and chemistry motivates us to develop the theory for constructing a fixed-width confidence interval for Gini index. 
%
\section{Appendix}
\begin{lemma}
Under the assumption that $\xi <\infty$, for any $d>0$, the stopping time $N_d$ is finite, that is, $P(N_d < \infty) =1$.
\end{lemma}
\begin{proof}
The lemma 1 is proved by using \eqref{stopping-rule} and the fact that $V_n^2$ is strongly consistent estimator of $\xi^2$ and $N_d \to \infty$ as $d \downarrow 0$ almost surely.
\end{proof}
\begin{lemma}\label{A1}
The value of sample Gini index lies between 0 and 1.
\end{lemma}
\begin{proof}
Let $Y_1,\dots,Y_n$ be the ordered incomes of $n$ persons where $Y_1$ represents the income of the poorest person and $Y_n$ represents the income of the richest person. Using \cite{damgaard2000describing}, Gini index can be rewritten as 
\begin{align}
0\leq\widehat{G}_n&=\frac{2\sum_{i=1}^{n}iY_i}{n\sum_{i=1}^{n}Y_i}-\frac{n+1}{n}\leq \frac{2n\sum_{i=1}^{n}Y_i}{n\sum_{i=1}^{n}Y_i}-\frac{n+1}{n}=\frac{n-1}{n}\leq 1.\notag
\end{align}
This proves the lemma.
\end{proof}
\subsection{Proof of Theorem \ref{thm:main}}
%
In this subsection, we prove some lemmas that are essential to establish theorem 1 and theorem 2. First, we introduce a few notations. Note from \eqref{stopping-rule} that $N_d \ge {\frac{z^2_{\alpha/2}}{d^2}} \, N_d^{-1}$, i.e., $N_d \ge {\frac{z_{\alpha/2}}{d}}(=m)$ with probability 1. Suppose $\boldsymbol{X}_{(n)} = (X_{(1)}, \ldots, X_{(n)}) $ denotes the $n$ dimensional vector of order statistics from the sample $X_1, \ldots, X_n$, and $\mathcal{F}_n$ is the $\sigma $-algebra generated by $( \boldsymbol{X}_{(n)}, X_{n+1}, X_{n+2}, \ldots).$ By \cite{lee1990u}, $\left\{ \overline{{ X}}_{n}{,\mathcal{F} }_{n}\right\} $, $\left\{ { S}_{{ n}}^{{ 2}}{ ,\mathcal{F} }_{n}\right\} $, $\left\{ \text{$\widehat{{ \tau}}_{n}{ ,\mathcal{F} }_{n}$}\right\} $, $\left\{ \widehat{{\Delta }}_{n},{ \mathcal{F} }_{n}\right\} $, and their convex functions are all reverse submartingales. 
%
Using reverse submartingale properties, let us prove the following lemmas.
\begin{lemma}\label{A2}Let $\overline{X}_n$ be the sample mean based on non-negative i.i.d. observations $X_1, \ldots, X_n$. Then, if $E(X^{-s})<\infty$, for $s>r$ and $r \ge 1$,
\begin{align}
E\( \max_{ n \ge m } \frac{1}{ \overline{X}_n^{r} } \)<\infty.
\end{align}
\end{lemma}
\begin{proof}
For $\alpha\ge 1$, we have
\begin{align}
E\( \max_{ n \ge m } \frac{1}{ \overline{X}_n^{r} } \)  \le
1+ \int_{1}^{\infty}P\( \max_{ n \ge m } \frac{1}{ \overline{X}_n^{r} } \ge t \) dt \le 1+ \frac{E\( \overline{X}_{m}^{-r\beta} \) }{\beta -1}, \label{A21}
\end{align}
where $\beta >1$. The last inequality is obtained by applying maximal inequality for reverse submartingales  (see \cite{lee1990u}). Let $s=r\beta$.  Now, it is enough to show that if $E(X^{-s})<\infty$, then $E\( \overline{X}_n^{\,\, -s} \) < \infty$. Note that $\overline{X}_n \ge \( \prod_{i=1}^n X_i \)^{1/n}$ as the observations are nonnegative and
\begin{align}
E\( \overline{X}_n^{\,\, -s} \) \le E\left[ \( \prod_{i=1}^n \frac{1}{X_i}\)^{s/n} \right] = \l E\left[ \( \frac{1}{X_1}\)^{s/n}\right]\r^n.
\label{A22}
\end{align} 
The last equality is due to the i.i.d. property of the observations. We know that $ \l E\( |X|^p\) \r^{1/p}$ is a nondecreasing function of $p$ for $p>0$. Applying this result with $p=1/n \le 1$ in \eqref{A22}, we complete the proof.
\end{proof}
\begin{lemma}\label{A3}
If $E(X_{1}^{4})$ and
$E(X_{1}^{-\beta})$ exist for $\beta > 4$, then $E\left[ \underset{n \ge m} \sup V_n^2 \right] < \infty$ for $m \ge 4$.
\end{lemma}
%
\begin{proof}
To prove lemma \ref{A3}, it is enough to show that: ${ E}\left[ \underset{n\geq m}{\sup } \,s_{wn}^{2}\overline{X}
_{n}^{-2}\right]$, ${ E}\left[
\underset{n\geq m}{\sup }\left\vert \frac{\widehat{\Delta }_{n}}{\overline{X}
_{n}^{3}}\widehat{\tau }_{n}\right\vert \right]$, ${ E}\left[ \underset{n\geq m}{\sup }\frac{\widehat{\Delta }
_{n}^{2}}{\overline{X}_{n}^{2}}\right]$, and  ${ E
}\left[ \underset{n\geq m}{\sup }\frac{\widehat{\Delta}_n^2}{\overline{X}_{n}^{4}}S_{n}^{2}\right]$ are finite. 
We note that, $0\leq\frac{\widehat{\Delta }_{n}}{2\overline{X}}\leq 1$. So, it is enough to show that ${ E}\left[ \underset{n\geq m}{\sup } \,s_{wn}^{2}\overline{X}_{n}^{-2}\right]$, ${ E}\left[\underset{n\geq m}{\sup } \frac{\widehat{\tau }_{n}}{\overline{X}_{n}^{2}}\right]$ and ${ E}\left[ \underset{n\geq m}{\sup }\frac{S_{n}^{2}}{\overline{X}_{n}^{2}}\right]$ are finite.  
Following Sen and Ghosh (p. 338, \cite{sen1981sequential}), we have $E\left[ \underset{n\geq m}{\sup }\, s_{wn}^4\right] < \infty$ if $E[X_1^{\alpha}] < \infty$ for $\alpha >4$ and $m \ge 4$. By lemma \ref{A2},  $E\left[ \underset{n\geq m}{\sup }\, \overline{X}_n^{\, -4} \right] < \infty$ if $E[X_1^{-\beta}] < \infty$ for $\beta >4$. Therefore,
\begin{align}
E\( \underset{n\geq m}{\sup } \, s_{wn}^2 \overline{X}_n^{\, -2} \) \le \l E\( \underset{n\geq m}{\sup } \, s_{wn}^4\) E\(\underset{n\geq m}{\sup }\, \overline{X}_n^{\, -4} \)  \r^{1/2} < \infty.
\end{align}
We note that $\widehat{\tau}_n$ and $S_n^2$ are U-statistics. Using lemma 9.2.4 of  \cite{ghosh1997sequential}, 
\[E\( \underset{n\geq m}{\sup }\left\vert \widehat{\tau}_n^2 \right\vert \) \le 4 E\( \left\vert \widehat{\tau}_m^2 \right\vert \)
\text{ and } E\( \underset{n\geq m}{\sup }\left\vert S_n^4 \right\vert \) \le  \left(\frac{4}{3}\right)^4E\( \left\vert S_m^4 \right\vert \).
\]
Applying Cauchy-Schwarz inequality,
\[
 E\( \underset{n\geq m}{\sup }\left\vert \frac{\widehat{\tau }_{n}}{\overline{X}_{n}^{2}}\right\vert \) 
\le \l E\( \underset{n\geq m}{\sup }\left\vert \widehat{\tau}_n^2 \right\vert \) \r^{\frac{1}{2}} \l E\( \underset{n\geq m}{\sup }\left\vert \overline{X}_n^{\,-4} \right\vert \) \r^{\frac{1}{2}} < \infty,
\]
and \[
 E\( \underset{n\geq m}{\sup }\left\vert \frac{S_{n}^2}{\overline{X}_{n}^{2}}\right\vert \) 
\le \l E\( \underset{n\geq m}{\sup }\left\vert S_n^4 \right\vert \) \r^{\frac{1}{2}} \l E\( \underset{n\geq m}{\sup }\left\vert \overline{X}_n^{\,-4} \right\vert \) \r^{\frac{1}{2}} < \infty,
\]
if $E(X_1^4)$ and $E(X_1^{-\beta})$ exist for $\beta > 4$. This completes the proof of lemma \ref{A3}.
\end{proof}
\noindent Below, we prove theorem \ref{thm:main} by using lemma \ref{A2} and \ref{A3}. 
\\ 

\emph{(i)} The definition of stopping rule $N_d$ in \eqref{stopping-rule} yields
\begin{align}\label{T1}
\left(\frac{z_{\alpha/2}}{d}\right)^2\, V_{N_{d}}^2 \, \le N_d  \, \le mI(N_d=m) \, + \, \left(\frac{z_{\alpha/2}}{d}\right)^2\left( { V}_{N_{d}-1}^2+{ (N}_{d}{ -1)}^{-1 }\right).
\end{align}
Since $N_d \to \infty$ a.s. as $d \downarrow 0$ and $V_n \to \xi$ a.s. as $n \to \infty$, by theorem 2.1 of Gut (2009), $V_{N_d}^2 \to \xi^2$ a.s..  Hence, dividing all sides of \eqref{T1} by $C$ and letting $d\downarrow 0$,
we prove $N_{d}/C\to 1$ a.s. as $d\downarrow 0$.

\noindent \emph{(ii)} Since $N_d \ge m$ a.s., dividing \eqref{T1} by $C$ yields
\begin{align}\label{T2}
N_d /C - mI(N_d=m)/C \le   \frac{1}{\xi^2} \( \underset{d >0}\sup \, V_{N_d -1}^2 + (m-1)^{-1}\) \,\,\text{ almost surely}.
\end{align}
Since $E\(\underset{d >0}\sup V_{N_d -1}^2 \) < \infty$ by lemma \ref{A3} and $N_{d}/C\to 1$ a.s. as $d\downarrow 0$, by the dominated convergence theorem, we conclude that  $\underset{d \downarrow 0}\lim E(N_d / C) =1$.
 This completes the proof of theorem \ref{thm:main}. 
 %
\subsection{Proof of Theorem \ref{thm:2}}
%
In order to show that our procedure satisfies the asymptotic consistency property, we will derive an Anscombe-type random central limit theorem for Gini index. This requires the existence of usual central limit theorem of Gini index and uniform continuity in probability (u.c.i.p.) condition. For details about the u.c.i.p. condition, we refer to \cite{anscombe1953sequential}, \cite{sproule1969sequential}, \cite{isogai1986asymptotic}, and \cite{mukhopadhyay2012tribute} etc.

First of all, let us define $n_1=(1-\rho)C$ and $n_2=(1+\rho)C$ for $0<\rho<1$. Now, we know from \cite{xu2007} that $\mathbf{Y}_n=\left(\sqrt{n}(\widehat{\Delta}_n-\Delta), \sqrt{n}(\bar{X}_n-\mu)\right)'\overset{\mathcal{L}}{\to}N_2(\mathbf{0},\Sigma)$, where
\[\Sigma =\left( \begin{array}{cc}
4\sigma_1^2 & 2(\tau-\mu\Delta)\\
2(\tau-\mu\Delta) & \sigma^2\\
\end{array} \right).\]
First, let us prove that $\mathbf{Y}_{N_d}\overset{\mathcal{L}}{\to}N_2(\mathbf{0},\Sigma)$.
Define $\mathbf{D}'=(a_0\text{  }a_1)$. Note that $\mathbf{D}'\mathbf{Y}_{N_d}=\mathbf{D}'\mathbf{Y}_C+(\mathbf{D}'\mathbf{Y}_{N_d}-\mathbf{D}'\mathbf{Y}_C)$. Thus, it is enough to show that $(\mathbf{D}'\mathbf{Y}_{N_d}-\mathbf{D}'\mathbf{Y}_C)\overset{P}{\to}0$ as $d\downarrow 0$. We can write
 \begin{align}
 \label{Deqn}
(\mathbf{D}'\mathbf{Y}_{N_d}-\mathbf{D}'\mathbf{Y}_C)  = & a_0\sqrt{N_d}(\widehat{\Delta}_N-\widehat{\Delta}_C)+a_1\sqrt{N_d}(\bar{X}_{N_d}-\bar{X}_C) \notag \\
&+(\sqrt{N_d/C}-1)\mathbf{D}'\mathbf{Y}_C.
\end{align}
Fix some $\epsilon >0$ and note that
\begin{align*}
&P\left\{\lvert a_0\sqrt{N_d}(\widehat{\Delta}_{N_d}-\widehat{\Delta}_C)+a_1\sqrt{N_d}(\bar{X}_{N_d}-\bar{X}_C)\rvert >\epsilon\right\}\notag\\
&\leq P\left\{\lvert a_0\sqrt{N_d}(\widehat{\Delta}_{N_d}-\widehat{\Delta}_C)+a_1\sqrt{N_d}(\bar{X}_{N_d}-\bar{X}_C)\rvert >\epsilon,\lvert N_d-C\rvert<\rho C\right\} \notag \\
&+P[\lvert {N_d}-C\rvert>\rho C]\notag\\
&\leq P\left\{\underset{n_1<n<n_2}{max}\lvert \sqrt{n}\lvert\widehat{\Delta}_n-\widehat{\Delta}_C\rvert>\frac{\epsilon}{2\lvert a_0\rvert}\right\}+P\left\{\underset{n_1<n<n_2}{max}\lvert \sqrt{n}\lvert\bar{X}_n-\bar{X}_C\rvert>\frac{\epsilon}{2\lvert a_1\rvert}\right\}\notag\\
& + P[\lvert {N_d}-C\rvert>\rho C]
\end{align*}
Here, $\widehat{\Delta}_n$ and $\bar{X}_n$ are both U-statistics which satisfy Anscombe's u.c.i.p. condition (for e.g. see \cite{sproule1969sequential}). Using u.c.i.p. condition and the fact that ${N_d}/C\overset{a.s.}{\to}1$, we conclude that for given $\epsilon >0$, there exist $\eta >0$ and $d_0 >0$ such that
\[
P\{\lvert a_0\sqrt{N_d}(\widehat{\Delta}_{N_d}-\widehat{\Delta}_C)+a_1\sqrt{N_d}(\bar{X}_{N_d}-\bar{X}_C)\rvert >\epsilon\}<\eta \quad \text {for all }\,\, d \le d_0.
\]
This implies $a_0\sqrt{N_d}(\widehat{\Delta}_{N_d}-\widehat{\Delta}_C)+a_1\sqrt{N_d}(\bar{X}_{N_d}-\bar{X}_C)  \overset{P}{\to} 0$ as $d \downarrow 0$. Also, note that $\left(\sqrt{N_d/C}-1\right)\mathbf{D}'\mathbf{Y}_C \overset{P}{\to} 0 $ as $d \downarrow 0$ since ${N_d}/C  \to 1$ almost surely and $\mathbf{D}'\mathbf{Y}_C  \overset{\mathcal{L}}{\to}N_2(\mathbf{0},\Sigma)$.
Thus, from \eqref{Deqn}, we conclude $(\mathbf{D}'\mathbf{Y}_{N_d}-\mathbf{D}'\mathbf{Y}_C)\overset{P}{\to}0$, that is, $\mathbf{Y}_{N_d}\overset{\mathcal{L}}{\to}N_2(\mathbf{0},\Sigma)$.
Now, define $G(u,v)=\frac{u}{2v}$, if $v\neq 0$. Using Taylor's expansion, we can write
\begin{align}
\sqrt{N_d}(G(\widehat{\Delta}_{N_d},\bar{X}_{N_d})-G(\Delta,\mu))&=\sqrt{N_d}\left(\frac{\widehat{\Delta}_{N_d} - \Delta}{2\mu} - \frac{\Delta}{2 \mu^2}(\overline{X}_{N_d} - \mu) + R_{N_d}\right),
\label{taylor}
\end{align}
where $R_{N_d} = -2(\widehat{\Delta}_{N_d} - \Delta)(\overline{X}_{N_d} - \mu)/ b^2 \, +\, 4 a (\overline{X}_{N_d} - \mu)^2/ b^3$, $a=\Delta + p(\widehat{\Delta}_{N_d} - \Delta)$, $b=2\mu + p(2\overline{X}_{N_d} - 2\mu)$, and $p \in (0, 1)$.
%
%
%
Rewriting (\ref{taylor}) in the vector-matrix form, we get
\begin{align}
\sqrt{N_d}(G(\widehat{\Delta}_{N_d},\bar{X}_{N_d})-G(\Delta,\mu))=\mathbf{D}'\mathbf{Y}_{N_d}+ \sqrt{N_d}R_{N_d},
\end{align}
where $\mathbf{D}'=\left(\frac{1}{2\mu}, \frac{-\Delta}{2\mu^2}\right)$. Note that $\sqrt{N_d} (\overline{X}_{N_d} - \mu)$ converges in distribution to a normal distribution by Anscombe's CLT and both $\left( \widehat{\Delta}_{N_d} - \Delta \right)$ and $\left(\overline{X}_{N_d} - \mu \right)$ converges to 0 almost surely. This yields $ \sqrt{N_d}R_{N_d} \overset{P}{\to} 0$ as $d \downarrow 0$.
 Hence, $\sqrt{N_d}(\widehat{G}_{N_d}-G_F)\overset{\mathcal{L}}{\to}N(\mathbf{0},D'\Sigma D)$  as $d\downarrow 0$. This completes the proof of theorem  \ref{thm:2}.

\newpage

\begin{table}
\centering
\caption{Performance of the proposed sequential procedure when the data is from Gamma, Log-normal, and Pareto}
\label{tab1}
\begin{tabular}{cccccc} \hline
&&&&& \\ [-1.2ex]
Distribution &$\,\;\,\overline{N}$ &$C$ &$\overline{N}/C$ &$\max (N)$ &$p$\\  [1ex]
&$\underset{}{\mathbf{s(}\overline{N}\mathbf{)}}$&&& &$s_p$\\ \hline%
&&&&& \\ [-1.2ex]
 Gamma & 1259.492 & 1267 & 0.9941 &1594 & 0.878\\[1ex]
&4.3639&&& &0.0073 \\ \hline%
&&&&& \\ [-1.2ex]
Log-normal &1429.349 & 1424 & 1.0038 & 2391 &0.9015\\
&4.1393&&& &0.0067\\ \hline%
&&&&& \\ [-1.2ex]
Pareto &654.5364 & 686 & 0.9541 & 1666 &0.9018\\
&4.2151&&& &0.0063\\ \hline%
\end{tabular}
\end{table}

\end{document}